\documentclass[11pt]{article}

\usepackage[left=2.5cm,right=2.5cm,top=2.5cm,bottom=2.5cm]{geometry}
\usepackage{graphicx,amssymb,amsmath,amsthm,mathrsfs,setspace,subcaption,cite,authblk,cancel,float}
\usepackage[all]{xy}

\usepackage[colorlinks=true,bookmarks=true]{hyperref}

\usepackage{mathtools}

\DeclarePairedDelimiterX\braket[2]{\langle}{\rangle}{#1 \delimsize\vert #2}

\hypersetup{citecolor=blue}
\newcommand{\dif}{\mathrm{d}}
\newcommand{\Eqref}[1]{(\ref{#1})}
\newcommand{\half}{\frac{1}{2}}

\newcommand{\brac}[1]{\left(#1 \right)}
\newcommand{\sbrac}[1]{\left[#1\right]}

\numberwithin{equation}{section}
 
\newtheorem{mythm}{Theorem}

\begin{document}

\title{Black hole thermodynamic free energy as $A$-discriminants}

\author[1]{Mounir Nisse\footnote{mounir.nisse@gmail.com, mounir.nisse@xmu.edu.my}}

\author[1]{Yen-Kheng Lim\footnote{yenkheng.lim@gmail.com, yenkheng.lim@xmu.edu.my}}

\author{Linus Chang\footnote{llhc2@srcf.net}}

\affil[1]{\normalsize{\textit{School of Mathematics and Physics, Xiamen University Malaysia, 43900 Sepang, Malaysia}}}

\date{\normalsize{\today}}
\maketitle 
 
\renewcommand\Authands{ and }
\begin{abstract}
 We show that the free energy $\mathcal{F}$ and temperature $T$ of black holes, considered as a thermodynamic system, can be viewed as an $A$-discriminant of an appropriately-defined polynomial. As such, mathematical results about $A$-discriminants may lead to implications about black hole thermodynamics. In particular, for static spacetimes with spherical, planar, or hyperbolic symmetry, the number of distinct thermodynamic phases depend on the number of distinct terms in the metric component $g_{tt}$. We prove that if $g_{tt}$ consists of $N_f$ distinct terms, then the $\mathcal{F}$-$T$ curve consists of $N_f-2$ cusps, which in turn leads to $N_f-1$ distinct thermodynamic phases. This result is applied to explicit examples of the Schwarzschild-AdS, Reissner--Nordstr\"{o}m, power-law Maxwell, and Euler--Heisenberg black holes.
\end{abstract}

\section{Introduction}

Various developments in past decades have demonstrated that black holes behave like thermodynamic systems. One such behaviour is that of phase transitions. In particular, Hawking and Page \cite{Hawking:1982dh} have shown that for a given range of temperatures, two possible phases of asymptotically Anti-de Sitter (AdS) black holes can exist. These two phases are typically called the \emph{large} and \emph{small} black hole branches due to their different horizon radii.  Furthermore, at temperatures higher than some critical value, the state of thermal AdS with no black hole is the thermodynamically favoured phase.

The Hawking--Page phase transition is of particular interest in the context of AdS/CFT duality, as the interpretation in the dual theory describes a confinement/deconfinement phase transition \cite{Witten:1998zw}. Since the works of Hawking and Page, similar thermodynamic analyses have been extended to various other spacetimes. For instance, the free energy $\mathcal{F}$ vs temperature $T$ curve for the the Reissner--Nordstr\"{o}m has the distinctive swallowtail feature similar to Van der Waals fluids \cite{Chamblin:1999tk,Chamblin:1999hg}. The analogy to Van der Waals have been completed by viewing the negative cosmological constant as a thermodynamic pressure \cite{Teitelboim:1985dp,Kastor:2009wy,Kubiznak:2012wp}. This viewpoint has be applied to black hole thermodynamics in other theories of gravity, such as Lovelock gravity \cite{Kastor:2010gq,Kastor:2011qp,Kastor:2016bph}. (For a review, see, e.g., \cite{Kubiznak:2016qmn,Kubiznak:2014zwa}.) This approach can be further applied to holography \cite{Kastor:2014dra} and further extended to scalar fields \cite{Kastor:2018cqc}. Now, the swallowtail curve of the free energy can be viewed as a curve with two cusps (i.e., singular points). In the Reissner--Nordstr\"{o}m case, these two cusps divide the curve into three parts, which we take to be distinct thermodynamic phases. 

In the same way, the $\mathcal{F}$ vs $T$ curve for the previously-mentioned Schwarzschild-AdS black hole consists of a single cusp which separates the `large' and `small' black hole branches. On the other hand, the $\mathcal{F}$ vs $T$ curve for the Euler--Heisenberg black holes \cite{Magos:2020ykt} may have up to three cusps, subdividing the curve into four distinct phases. If one consideres matter fields coming from general non-linear electrodynamics, the corresponding curves can have even higher number of cusps \cite{Tavakoli:2022kmo,Quijada:2023fkc}. As was argued in \cite{Tavakoli:2022kmo}, the parameters can be adjusted such that the cusps coincide, which result in generalisations of thermodynamic triple points, namely are quadruple points or generally $n$-tuple points. From these we observe that various black holes with different matter fields seem to have different number of cusps in their $\mathcal{F}$-$T$ curves, and hence different number of phases. Therefore it seems worthwhile to ask: \emph{What dictates the number of distinct thermodynamic phases of a given spacetime?}

The aim of this paper is to approach this question by viewing the $\mathcal{F}$ vs $T$ curve as an $A$-discriminant. The theory of $A$-discriminants is well established in the mathematical literature \cite{GKZ-94}. In the present paper we will use it to understand why different black hole solutions exhibit different number of phases. In particular, we will see how the number of distinct terms in the metric component $g_{tt}$ dictates the number distinct phases in its thermodynamics. This is done by establishing the upper bound in the number of cusps of the $A$-discriminants, which translates to cusps in the $\mathcal{F}$-$T$ curve. It is then natural to ask whether there exist any results in the mathematical literature that may assist in addressing this problem. To our knowledge, the most relevant existing result is one by Dickenstein et al. \cite{Dickenstein2007}, subsequently sharpened in Rusek's thesis \cite{Rusekthesis}, which (loosely speaking) states that for a polynomial of $n$ variables with $n+3$ terms, its $A$-discriminant has at most $n$ cusps. For black hole thermodynamics, we only have one relevant variable which is the horizon radius $r_+$. Therefore $n=1$ and this theorem states that a polynomial of 4 terms can has at most 1 cusp. This statement is only applicable to the Schwarzschild-AdS black hole, where its relevant polynomial indeed is a sum of exactly 4 terms and its thermodynamic curve has 1 cusp separating the large and small black hole branches. To consider other systems (for instance the Reissner--Nordstr\"{o}m and Euler--Heisenberg black holes), we have to extend this result to univariate polynomials with an arbitrary number of terms. If one considers the grand canonical ensemble where the electric potential of the spacetime is fixed at the boundary, then the black hole charge becomes another extensive parameter. In this case, the problem is now described by a bivariate polynomial. In previous literature, the grand canonical ensemble for the Reissner--Nordstr\"{o}m-AdS black hole is considered in, for example, Refs.~\cite{Chamblin:1999tk,Chamblin:1999hg,Ghosh:2021uxg}.

The rest of this paper is organised as follows. In Sec.~\ref{sec_thermo} we review the neccessary notions in black hole thermodynamics and establish an alternative viewpoint of the first law of thermodynamics. This alternate view is shown in Sec.~\ref{sec_discrim} to be the $A$-discriminant of a univariate polynomial. In the same section, we give an upper bound of the number of cusps in a two-dimensional cross section of the $A$-discriminant. Explicit examples are studied in Sec.~\ref{sec_examples}. In Sec.~\ref{sec_gc}, we consider the grand canonical ensemble for the Reissner--Nordstr\"{o}m-AdS black hole. The paper concludes with Sec.~\ref{sec_conclusion}. Here, we work in units where $c=\hbar=k_{\mathrm{B}}=1$ and our convention for Lorentzian signature is $(-,+,+,\ldots,+)$.

\section{Thermodynamic quantities and the first law} \label{sec_thermo}

In this section, we shall review the basic notions of black hole thermodynamics, particularly its free energy $\mathcal{F}$, temperature $T$, and other thermodynamic quantities. We focus particularly on static spacetimes with spherical, planar, or hyperbolic symmetry. 

To start, let us consider $d$-dimensional spacetimes of the form 
\begin{align}
 \dif s^2&=-f(r)\dif t^2+f(r)^{-1}\dif r^2+r^2\gamma_{ij}\dif\theta^i\dif\theta^j, \label{general_metric}
\end{align}
where $\gamma_{ij}\dif\theta^i\dif\theta^j$ is the metric on $S^{d-2}_k$, a $(d-2)$-dimensional space of constant curvature $k=0,\pm 1$. These are solutions for which $g_{tt}=-1/g_{rr}$ \cite{Jacobson:2007tj}, which typically describe spacetimes in the presence of (non-)linear  electromagnetic fields, or when the stress-energy tensor components satisfy ${T^t}_t={T^r}_r$. Spacetimes of this form was considered by Padmanabhan to describe a local notion of thermodynamic energy \cite{Padmanabhan:2002sha}. 

We further assume $f(r)$ takes the form 
\begin{align}
 f(r)=k-\frac{\mu}{r^{d-3}}+Y(r). \label{f_form}
\end{align}
So $Y(r)=0$ solves the vacuum equation $R_{\mu\nu}=0$. (This is the Schwarzschild solution.) The presence of non-zero $T_{\mu\nu}$ will contribute additional terms to $Y(r)$, where we assume to be a sum of distinct powers of $r$ consisting of $K$ terms. Specific examples of $Y(r)$ will be considered in Sec.~\ref{sec_examples}. In order for the theory of $A$-discriminants to be applicable, we require $Y(r)$ to a rational function. 

The horizon will be taken as $r=r_+$, the largest root of $f(r)$, such that 
\begin{align}
 f(r_+)=0=k-\frac{\mu}{r_+^{d-3}}+Y(r_+)\quad\leftrightarrow\quad\mu=kr_+^{d-3}+r_+^{d-3}Y(r_+).
\end{align}
The temperature can be obtained using the Euclidean periodicity trick \cite{Hartle:1976tp,Gibbons:1976pt,Gibbons:1976ue,Lewkowycz:2013nqa} which leads to
\begin{align}
 T=\frac{f'(r_+)}{4\pi}=\frac{1}{4\pi r_+}\sbrac{(d-3)k+(d-3)Y(r_+)+r_+Y'(r_+)}.\label{general_temperature}
\end{align}
The entropy is defined as one quarter of the horizon area,
\begin{align}
 S=\frac{\Omega r_+^{d-2}}{4G}, \label{general_entropy}
\end{align}
where $\Omega=\int\dif^{d-2}\theta\sqrt{\gamma}$ is the area of $S^{d-2}_k$. In particular, for $k=1$, it is the area of the $(d-2)$-sphere, which is $\Omega={2\pi^{\frac{d-1}{2}}}/{\Gamma\brac{\frac{d-1}{2}}}$. Let us further assume that the mass of the black hole is computed to be 
\begin{align}
 M=\frac{\Omega}{16\pi G}(d-2)\mu=\frac{\Omega}{16\pi G}(d-2)\sbrac{kr_+^{d-3}+r^{d-3}_+Y(r_+)}. \label{general_mass}
\end{align}
In other words, fixing a horizon radius $r_+$ determines the black hole's mass $M$, temperature $T$, and entropy $S$.  In the canonical ensemble, we take the free energy of this generic spacetime to be $\mathcal{F}=M-TS$. 

Now, viewing it in a slightly different manner, suppose we define a function $\mathcal{P}=\mathcal{F}-M+TS$. This vanishes whenever $\mathcal{F}$ takes the value of the free energy of the system,  
\begin{align}
 \mathcal{P}=\mathcal{F}-M+TS=0. \label{condition1}
\end{align}
If we consider a neighbouring solution with a slightly perturbed horizon radius 
\begin{align}
 r_+\rightarrow r_++\delta r_+,
\end{align}
then the mass and entropy of this new solution differs according to
\begin{align}
 M\rightarrow M+\delta M,\quad S\rightarrow S+\delta S,
\end{align}
where
\begin{align}
 \delta M=\frac{\partial M}{\partial r_+}\delta r_+,\quad \delta S=\frac{\partial S}{\partial r_+}\delta r_+.
\end{align}
The derivatives $\frac{\partial M}{\partial r_+}$ and $\frac{\partial S}{\partial r_+}$ can be obtained by directly differentiating Eqs.~\Eqref{general_mass} and \Eqref{general_entropy}. At fixed temperature, one can directly show that this obeys the first law of thermodynamics 
\begin{align}
 \delta M=T\delta S. \label{FirstLaw} 
\end{align}
Taking $\delta\mathcal{F}=0$, the first law \Eqref{FirstLaw} can then be written as 
\begin{align}
 \delta\brac{\mathcal{F}-M+TS}&=\frac{\partial}{\partial r_+}\brac{\mathcal{F}-M+TS}\delta r_+=\mathcal{P}'\delta r_+=0.\label{condition2}
\end{align}

Strictly speaking, for the types of spacetime considered here $\mathcal{P}$ is a rational function which can be written in the form 
\begin{align}
 \mathcal{P}=\frac{P}{H},
\end{align}
where $P$ and $H$ are polynomials in $r_+$. Assuming $H$ is non-vanishing, Eqs.~\Eqref{condition1} and \Eqref{condition2} is equivalent to 
\begin{align}
 P=P'=0, \label{condition_discriminant}
\end{align}
Viewed from this perspective, the thermodynamics of spacetimes of the form \Eqref{general_metric} is determined from its mass $M$ and entropy $S$. The free energy $\mathcal{F}$ and temperature $T$ is subsequently determined from the vanishing of $P$ and its derivative. Note that $T$ has been obtained in this manner as an alternative to, and agrees with, the Euclidean periodicity trick.

\section{\texorpdfstring{$A$}{A}-discriminants, cusps, and thermodynamic phases} \label{sec_discrim}

In the language of algebraic geometry, the condition \Eqref{condition_discriminant} is what defines an $A$-discriminant. More precisely, consider univariate polynomials of the form 
\begin{align}
 P(x)=c_1 x^{a_1}+c_2 x^{a_2}+\ldots+c_N x^{a_N},
\end{align}
where $A=\left\{a_1,a_2,\ldots,a_N \right\}$ is a set of $N$ integer exponents of the polynomial, called the \emph{support}, and $c_1,c_2,\ldots,c_N$ are the coefficients of the polynomial. Since we are mainly concerned with roots $P(x)=0$, integer multiples of coefficients are considered to be equivalent. That is, $(c_1,c_2,\ldots,c_N)\sim(\lambda c_1,\lambda c_2,\ldots,\lambda c_N)$ for non-zero $\lambda$. In other words, we can parametrise the family of polynomials geometrically by viewing the coefficients as elements of projective space $[c_1:c_2:\ldots:c_N]\in\mathbb{R}P^{N-1}$. In this way, any polynomial of a given support $A$ is represented by a point in $\mathbb{R}P^{N-1}$. Given a support set $A$, the subset of polynomials that satisfies
\begin{align}
 P(x)=P'(x)=0 \label{def_Adiscriminant}
\end{align}
defines a co-dimension-1 hypersurface in $\mathbb{R}P^{N-1}$. This hypersurface is called the $A$-\emph{discriminant}.\footnote{Note that $A$-\emph{discriminants} are related to the conventional notion of \emph{discriminants}, but not identical to it. This can be shown using the elementary example $P(x)=ax^2+bx+c$: In this case, the set of all quadratic polynomials can be represented by points $[a:b:c]\in\mathbb{R}P^2$. The \emph{discriminant} is $\Delta(P)=b^2-4ac$, whereas the $A$-\emph{discriminant} is defined by the vanishing of the discriminant $b^2-4ac=0$, a 1-dimensional curve in $\mathbb{R}P^2$.} 

 Determining the singular points of an $A$-discriminant is typically a non-trivial problem, especially in the case of multivariable polynomials. In a typical treatment of these problems, one starts with a configuration of $N$ lattice points $A \subset \mathbb{Z}^n$, let  \begin{align}
 F_A(z_1,\ldots,z_n)=\sum_{j=1}^N c_j z_1^{a_1^j}z_2^{a_2^j}\cdots z_n^{a_n^j},
\end{align}
denote the generic polynomial in $n$ complex variables $(z_1, . . . , z_n)$ with exponents in $A$.  Under certain general conditions, Gelfand, Kapranov and Zelevinsky \cite{GKZ-94} showed that there exists an irreducible polynomial with integer coefficients $D_A = D_A(z)$ in the vector of coefficients $a^j=(a_1^j,a_2^j,\ldots,a_n^j)\in A$ (defined up to sign), called the $A$-discriminant, which vanishes for each choice of coefficients $c$ for which $F_A$ and all its partial derivatives have a common root in the torus $(\mathbb{C}^*)^n$. The $A$-discriminant is an affine invariant, in the sense that any configuration of points affinely isomorphic to $A$ has the same discriminant.

A theorem by Dickenstein et al., gives a particular property of $A$-discriminants based on its polynomial, stated here in a form modified to our context,
\begin{mythm}[Dickenstein et al. (2007) \cite{Dickenstein2007}, Rusek (2013) \cite{Rusekthesis}]\label{thm_korben}
 For polynomials of $n$ variables consisting of $n+3$ terms (monomials), its reduced $A$-discriminant has at most $n$ cusps.
\end{mythm}

For black hole thermodynamics, we only have a single relevant variable $r_+$. So $n=1$. Therefore this theorem only states that the $A$-discriminant of polynomial of 4 terms has at most 1 cusp. This theorem is only applicable to a single case, namely the Schwarzschild-AdS black hole. (See Sec.~\ref{sec_examples_SchAdS}) In this form, this theorem is not able to explain the number of cusps in the thermodynamics of other black holes. 

Restricting our attention to univariate polynomials, it is desireable to have an extension of Theorem \ref{thm_korben} to predict or explain the number of cusps for other solutions. To obtain such an extension, we note that, as seen in Sec.~\ref{sec_thermo}, the free energy $\mathcal{F}$ and temperature $T$ appear as two particular coefficients of a polynomial. Hence let us consider two-dimensional cross sections of an $A$-discriminant. In particular, let us write the polynomial as 
\begin{align}
 P(x)=ax^m+bx^n+\varphi(x), \label{discrim_P}
\end{align}
where $\varphi(x)$ is a polynomial of $M=N-2$ terms. Contact to thermodynamics is made by identifying $a$ and $b$ to be the temperature and free energy, respectively.  The condition $P(x)=P'(x)=0$ leads to 
\begin{align}
 a=\frac{x\varphi'-n\varphi}{(n-m)x^m},\quad b=\frac{x\varphi'-m\varphi}{(m-n)x^n}. \label{gamma_curve}
\end{align}

Let us denote by $\gamma_A$ the cross section of the $A$-discriminant obtained as a curve $(a(x),b(x))$ parametrised by $x$. This curve may be singular at points when the derivatives of $a$ and $b$ with respect to $x$ vanishes simultaneously. We refer to these points as \emph{cusps}. For a univariate polynomial, the following argument gives an upper bound of the number of cusps on $\gamma_A$:
\begin{mythm} \label{thm_unicusps}
 Let $P(x)$ be a real univariate polynomial as given in Eq.~\Eqref{discrim_P} which consists of $N$ terms in the domain $x>0$. The cross section $\gamma_A$ of the $A$-discriminant has at most $N-3$ cusps.
\end{mythm}

\begin{proof}
 The cusps are the points where $a'=b'=0$. Evaluating the derivatives explicitly, we find 
 \begin{align*}
 a'=\frac{x^2\varphi''+(1-m-n)x\varphi'+mn\varphi}{(n-m)x^{m+1}},\quad b'=\frac{x^2\varphi''+(1-m-n)x\varphi'+mn\varphi}{(m-n)x^{n+1}}.
\end{align*}
So the cross section of the $A$-discriminant has cusps when their common numerators of $a'$ and $b'$ vanishes,
\begin{align}
 x^2\varphi''+(1-m-n)x\varphi'+mn\varphi=0. \label{apbp_numerator}
\end{align}
Note that the left-hand side of Eq.~\Eqref{apbp_numerator} is itself another polynomial. The number of roots of this polynomial therefore gives the number of cusps of $\gamma_A$. Now, if $\varphi(x)$ is a polynomial of $M$ terms, then the left hand side of \Eqref{apbp_numerator} will also have at most $M$ terms. By the Descartes Rule of Signs, it can have at most $(M-1)$ positive roots. Since $N=M+2$, the curve $\gamma_A$ will have at most $N-3$ cusps.
\end{proof}

In the context of black hole thermodynamics, $x=r_+$ is the horizon radius (up to constant factors and/or powers), which was the reason why the domain $x>0$ is assumed in Theorem \ref{thm_unicusps}. To see how this mathematical result determines the number of thermodynamic phases, we observe that the $\mathcal{F}$-$T$ curve is obtained form the numerator of $\mathcal{F}-E+TS$. If the spacetime in Eq.~\Eqref{general_metric} has $f(r)=k-\frac{\mu}{r^{d-3}}+Y(r)$ which consists of a sum of $N_f$ distinct powers of $r$, then its expression of mass is a sum of $N_f-1$ distinct powers of $r_+$. (See Eq.~\Eqref{general_mass}.) The entropy is simply the single-term polynomial $\frac{\Omega}{4G}r_+^{d-2}$. Then, the numerator of $\mathcal{F}-E+TS$ is a polynomial of $N_f+1$ terms. If we take $\gamma_A$ to be $\brac{T(r_+),\mathcal{F}(r_+)}$, then by the results discussed above, it will have 
\begin{align}
 N_{\mathrm{cusps}}=N-3=N_f-2\quad\mbox{cusps}. \label{N_cusps}
\end{align}
In the $T$-$\mathcal{F}$ plane, the cusps separate distinct thermodynamic phases. Hence there will be $N_f-1$ thermodynamic phases.

Before closing this section, let us make some remarks about the ambient space in which $A$-discriminant lives, in relation to black-hole thermodynamics. Recalling that Eq.~\Eqref{condition_discriminant} follows from the first law and the definition of the free energy, and this coincides precisely with Eq.~\Eqref{def_Adiscriminant} which defines the $A$-discriminant. As discussed above, this $A$-discriminant is a co-dimension 1 hypersurface living in $\mathbb{R}P^{N-1}$, which is the projective space parametrising the coefficients of all polynomials with support $A$. 

A point in $\mathbb{R}P^{N-1}$ is the `vector' of polynomial coefficients written as $(c_1,c_2,\ldots,c_N)$. Some of these entries characterise the thermodynamic quantities of the black hole, while others are simply numerical factors. At least two of them are the temperature and free energy respectively. Say, $c_1\propto T$ and $c_2\propto\mathcal{F}$. In the grand canonical ensemble, a third intensive parameter enters, $c_3\propto\Phi$ where $\Phi$ is the electric potential, whereas the relevant thermodynamic potential should be the grand potential $\mathcal{W}$. That is,  $c_2\propto\mathcal{W}$. In other words, the physically relevant (thermodynamic) space of the $A$-discriminant is the cross-section where one axis is the thermodynamic potential while the rest are intensive parameters. This perhaps can be related to a simple case of the thermodynamic phase space of \cite{Rajeev:2007uk,Ghosh:2019rsu}, where the black hole mass $M$, temperature $T$, and entropy $S$ forms a three-dimensional contact manifold. Here, we take the canonical ensemble by considering the Legendre transform $\mathcal{F}$ of $M$, and the basic variables are $(\mathcal{F},T,S)$. Equilibrium states form a submanifold in this space, and the $A$-discriminant is the one-dimensional slice in the $\mathcal{F}$-$T$ plane.

\section{Examples} \label{sec_examples}

In this section we compute the $\mathcal{F}$-$T$ curves explicitly for various black hole solutions. The examples below are known results in the literature, but we use them as particular cases that verify the statement of Theorem \ref{thm_unicusps}.

\subsection{Schwarzschild-AdS} \label{sec_examples_SchAdS}

The Schwarzschild-Anti-de Sitter black hole, is described by the metric
\begin{subequations}\label{metric_AdSBH}
\begin{align}
 \dif s^2&=-f(r)\dif t^2+f(r)^{-1}\dif r^2+r^2\dif\Omega_{(d-2)}^2,\\
  f(r)&=1-\frac{\mu}{r^{d-3}}+\frac{r^2}{\ell^2},\quad \ell^2=-\frac{(d-1)(d-2)}{2\Lambda}.
\end{align}
\end{subequations}
This is the metric of the form \Eqref{general_metric} for which $k=1$. So we use the notation where $\dif\Omega_{(d-2)}^2$ is the line element of a unit sphere $S^{d-2}$. This is a solution to Einstein's equation $R_{\mu\nu}=\frac{2\Lambda}{d-2}g_{\mu\nu}$ where $\Lambda$ is a negative cosmological constant. We see that $f(r)$ here consists of $N_f=3$ terms. The corresponding function $\mathcal{P}=\mathcal{F}-M+TS$ then consists of exactly 4 terms. This is the only case for which Theorem \ref{thm_korben} by Dickenstein et al. \cite{Dickenstein2007} is applicable, which states that its $\mathcal{F}$-$T$ curve has at most $n=1$ cusp. Nevertheless, it is also a case applicable to Theorem \ref{thm_unicusps} and Eq.~\Eqref{N_cusps}. Indeed, this single cusp separates the large and small black hole branches \cite{Hawking:1982dh}. 

We now obtain its temperature and free energy using the $A$-discriminant calculation. The mass and entropy of this spacetime is 
\begin{subequations}
\begin{align}
 M&=\frac{(d-2)\Omega}{16\pi G}\mu=\frac{(d-2)\Omega}{16\pi G}\brac{r_+^{d-3}+\frac{r_+^{d-1}}{\ell^2}},\\
 S&=\frac{\mathcal{A}}{4G}=\frac{\Omega}{4G}r_+^{d-2}
\end{align}
\end{subequations}
Taking $\mathcal{P}=\mathcal{F}-M+TS$, we have, upon rearranging,
\begin{align}
 \mathcal{P}&=\frac{\Omega}{16\pi G}\sbrac{-(d-2)r_+^{d-3}-(d-2)\frac{r_+^{d-1}}{\ell^2}+(4\pi T)r_+^{d-2}+\frac{16\pi G}{\Omega}\mathcal{F}}.
\end{align}
Introducing dimensionless variables $x=\frac{r_+}{\ell}$, $a=4\pi\ell T$, and $b=\frac{16\pi G}{\Omega \ell^{d-3}}\mathcal{F}$, we have 
\begin{align}
 \frac{16\pi G}{\Omega\ell^{d-3}}\mathcal{P}=P(x)=ax^{d-2}+b-(d-2)x^{d-3}-(d-2)x^{d-1},
\end{align}
so that $P(x)$ is a polynomial of degree $d-1$ with yet-to-be-determined coefficients $a$ and $b$. The $A$-discriminant of $P(x)$ is obtained by $P(x)=P'(x)=0$. Solving the simultaneous equations for $a$ and $b$ gives 
\begin{align}
 a=\frac{(d-3)x^{d-3}+(d-1)x^{d-1}}{x^{d-2}},\quad b=x^{d-3}-x^{d-1}.
\end{align}
Restoring the dimensionful quantities we recover 
\begin{align}
 T&=\frac{1}{4\pi r_+}\brac{d-3+(d-1)\frac{r_+^2}{\ell^2}},\quad \mathcal{F}=\frac{\Omega}{16\pi G}\brac{r_+^{d-3}-\frac{r_+^{d-1}}{\ell^2}}.
\end{align}
As expected, this expression for $T$ agrees with the Euclidean periodicity trick, and the free energy agrees with the calculation obtained by evaluating the Euclidean action on-shell, $I=\beta\mathcal{F}$ where $\beta=1/T$.

The plot of $\mathcal{F}$ against $T$ is shown in Fig.~\ref{fig_SchAdS}. We shall view these as a parametrised curve $\brac{T(r_+),\mathcal{F}(r_+)}$. The curve has a single cusp (i.e., a singular point where the tangent to the curve is undefined), corresponding to $r_+=r_1\equiv\sqrt{\frac{d-3}{d-1}}\;\ell$ and this cusp subdivides the curve into two segments. The branch with $r_+>r_1$ is called the \emph{large black hole} branch, marked as \textsf{A} in Fig.~\ref{fig_SchAdS} and the one with $r_+<r_1$ is the \emph{small black hole} branch, marked as \textsf{B} in Fig.~\ref{fig_SchAdS}. As mentioned above, $f(r)$ consists of $N_f=3$ terms. By the result of the previous section, this leads to the free energy curve having $N_f-2=1$ cusp.
\begin{figure}
 \centering
 \includegraphics{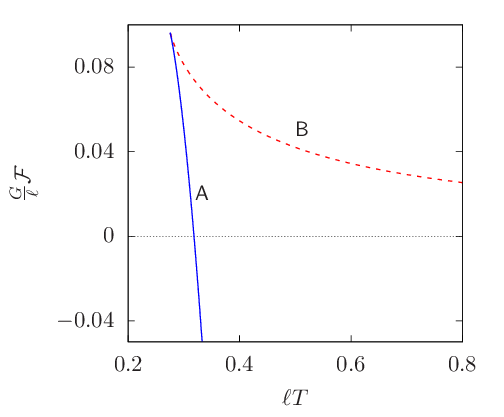}
 \caption{Graph of free energy $\mathcal{F}$ vs temperature $T$ for the spherically-symmetric Schwarzschild-AdS black hole. Branch \textsf{A} corresponds to $r_1<r_+<\infty$ and branch \textsf{B} corresponds to $0<r_+<r_1$, where $r_1=\sqrt{\frac{d-3}{d-1}}\;\ell$.}
 \label{fig_SchAdS}
\end{figure}

\subsection{Reissner--Nordstr\"{o}m-AdS}

Now consider Einstein--Maxwell theory described by the action 
\begin{align}
 I=\frac{1}{16\pi G}\int\dif^dx\sqrt{-g}\brac{R-2\Lambda-F_{\mu\nu}F^{\mu\nu}},
\end{align}
where $F=\dif A$ is the Faraday tensor which is the exterior derivative of the gauge potential 1-form $A$. In components, $F_{\mu\nu}=\nabla_\mu A_\mu-\nabla_\nu A_\nu$. Extremising the action leads to the equations of motion
\begin{subequations}
 \begin{align}
  &R_{\mu\nu}=\frac{2\Lambda}{d-2}g_{\mu\nu}+2F_{\mu\lambda}{F_\nu}^\lambda-\frac{1}{d-2}F_{\rho\sigma}F^{\rho\sigma}g_{\mu\nu},\\
  &\nabla_\mu F^{\mu\nu}=0.
 \end{align}

\end{subequations}
The Reissner--Nordstr\"{o}m solution describes a charged black hole that is asymptotically AdS. The metric and gauge potential are given by
\begin{subequations}\label{metric_RNAdS}
\begin{align}
 \dif s^2&=-f(r)\dif t^2+f(r)^{-1}\dif r^2+r^2\gamma_{ij}\dif\theta^i\dif\theta^j,\\
  f(r)&=k-\frac{\mu}{r^{d-3}}+\frac{q^2}{r^{2(d-3)}}+\frac{r^2}{\ell^2},\quad \ell^2=-\frac{(d-1)(d-2)}{2\Lambda},\\
  A&=\brac{\Phi-\sqrt{\frac{d-2}{2(d-3)}} \frac{q}{r^{d-3}}}\dif t,
\end{align}
\end{subequations}
where $\mu$ and $q$ are the mass and charge parameter, respectively. Here $k$ is the unit curvature parameter for the metric $\gamma_{ij}$. For this solution we have $N_f=4$ terms. This lies beyond the assumptions of Theorem \ref{thm_korben} by Dickenstein et al. \cite{Dickenstein2007}. However, by Theorem \ref{thm_unicusps}, its $\mathcal{F}$-$T$ curve should have at most 2 cusps. These two cusps gives the swallowtail figure studied in Refs.~\cite{Chamblin:1999tk,Chamblin:1999hg,Kubiznak:2012wp}.

We now obtain its temperature and free energy using the $A$-discriminant. As before, the horizon radius is denoted by $r_+$, for which $f(r_+)=0$, so that
\begin{align}
 \mu=kr_+^{d-3}+\frac{r_+^{d-1}}{\ell^2}+\frac{q^2}{r^{d-3}_+}. \label{mu_to_rplus}
\end{align}
We will take $(r_+,q)$ to parametrise the family of RNAdS solutions. For a given $r_+$, we can recover $\mu$ using Eq.~\Eqref{mu_to_rplus}.

The mass, charge, and entropy are given by
\begin{subequations} \label{RNAdS_quantities}
\begin{align}
 M&=\frac{\Omega}{16\pi G}(d-2)\mu=\frac{\Omega}{16\pi G}(d-2)\brac{kr_+^{d-3}+\frac{r_+^{d-1}}{\ell^2}+\frac{q^2}{r_+^{d-3}}},\\
 Q&=\frac{\Omega}{16\pi G}\sqrt{8(d-2)(d-3)}\;q,\\
 S&=\frac{\Omega}{4G}r_+^{d-2}.
\end{align}
\end{subequations}
In this case, the function $\mathcal{P}=\mathcal{F}-M+TS$ is
\begin{align}
 \mathcal{P}&=\mathcal{F}-\frac{\Omega}{16\pi G}(d-2)\brac{kr_+^{d-3}+\frac{r_+^{d-1}}{\ell^2}+\frac{q^2}{r^{d-3}_+}}+T\cdot\frac{\Omega}{4G}r_+^{d-2}\nonumber\\
 \frac{16\pi G}{\Omega\ell^{d-3}}x^{d-3}\mathcal{P}&=bx^{d-3}-(d-2)\brac{kx^{2d-6}+x^{2d-4}+w^2}+ax^{2d-5},\label{RNAdS_Pcal}
\end{align}
where we have defined dimensionless quantities
\begin{align}
 x&=\frac{r_+}{\ell},\quad w=\frac{q}{\ell^{d-3}},\quad a=4\pi\ell T,\quad b=\frac{16\pi G}{\Omega\ell^{d-3}}\mathcal{F}.
\end{align}
The thermodynamic curve is obtained as the discriminant of the right hand side of \Eqref{RNAdS_Pcal}, which is a polynomial 
\begin{align}
 P(x)=ax^{2d-5}+bx^{d-3}-(d-2)\brac{kx^{2d-6}+x^{2d-4}+w^2}
\end{align}
Computing curve $\gamma_A$, we have 
\begin{align}
 a&=\frac{(d-3)kx^{2d-6}+(d-1)x^{2d-4}-(d-3)w^2}{x^{2d-5}},\quad b=\frac{kx^{2d-6}-x^{2d-4}+(2d-5)w^2}{x^{d-3}}
\end{align}
Restoring dimensionful variables, we recover 
\begin{align}
 T=\frac{1}{4\pi r_+}\sbrac{(d-3)k+(d-1)\frac{r_+^2}{\ell^2}-(d-3)\frac{q^2}{r_{+}^{2d-6}}},\quad\mathcal{F}=\frac{\Omega}{16\pi G}\brac{kr_+^{d-3}-\frac{r_+^{d-1}}{\ell^2}+(2d-5)\frac{q^2}{r_+^{d-3}}}.
\end{align}
This agrees with the Euclidean periodicity trick, along with the Euclidean action $I=\beta\mathcal{F}$ evaluated with fixed charge boundary conditions. 

\begin{figure}
 \centering
 \includegraphics{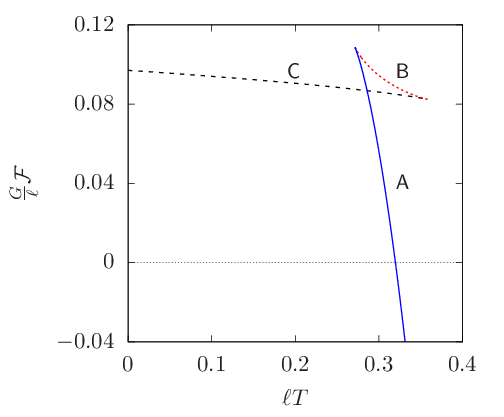}
 \caption{Graph of free energy $\mathcal{F}$ against temperature $T$ for the Reissner--Nordstr\"{o}m-AdS black hole in the case $k=1$, $d=4$, and $q=0.09667\ell$. In this case we have the distinctive swallowtail figure which has two cusps subdividing the curve into three segments, \textsf{A} ($r>r_2$), \textsf{B} ($r_1<r<r_2$), and \textsf{C} ($0<r_+<r_1$), where $r_1$ and $r_2$ are as given in Eq.~\Eqref{RNAdS_cusps}.}
 \label{fig_RNAdS}
\end{figure}

In Fig.~\ref{fig_RNAdS}, we plot the curve for the case $d=4$, $k=1$, and $q=0.09667\ell$, which describes the spherically symmetric charged black hole in AdS. 
As was already pointed out in \cite{Chamblin:1999tk}, we have the swallowtail figure similar to the Van der Waals fluid. In terms of $A$-discriminants, we have in this case $Y(r)=\frac{r^2}{\ell^2}+\frac{q^2}{r^{2d-6}}$, a sum of two terms. Therefore by the result of the previous section, it explains why the free energy curve has two cusps which then leads to this swallowtail figure. In this case, there are two cusps when $r_+$ takes the values
\begin{align}
 r_1=0.1757761483\ell,\quad r_2=0.5499418869\ell,\label{RNAdS_cusps}
\end{align}
which divides the curve into three segments, $0<r_+<r_1$ (segment \textsf{A}), $r_1<r<r_2$ (segment \textsf{B}), and $r_+>r_1$ (segment \textsf{C}).

\subsection{Einstein-power-Maxwell}

We now consider an Einstein--Maxwell theory where the Maxwell term $F_{\mu\nu}F^{\mu\nu}$ has an exponent $p$ that is not neccessarily 1. Theories of this form was considered in considered in \cite{Gonzalez:2009nn,Roychowdhury:2012vj}. This is a non-linear generalisation of Einstein--Maxwell theory, whose action is 
\begin{align}
 I=\frac{1}{16\pi G}\int\dif^dx\sqrt{-g}\brac{R-2\Lambda+X^p},
\end{align}
where $X=-F_{\mu\nu}F^{\mu\nu}$. The equations of motion are 
\begin{subequations}
\begin{align}
 &R_{\mu\nu}=\frac{2\Lambda}{d-2}g_{\mu\nu}+X^{p-1}\brac{2pF_{\mu\lambda}{F_\nu}^\lambda+\frac{2p-1}{d-2}Xg_{\mu\nu}},\\
 &\nabla_\mu\brac{X^{p-1}F^{\mu\nu}}=0.
\end{align}
\end{subequations}
These equations admit a black hole solution 
\begin{subequations}
 \begin{align}
  \dif s^2&=-f(r)\dif t^2+f(r)^{-1}\dif r^2+r^2\gamma_{ij}\dif\theta^i\dif\theta^j,\\
  f(r)&=k-\frac{\mu}{r^{d-3}}+\frac{r^2}{\ell^2}+\frac{q^2}{r^c},\quad \ell^2=-\frac{(d-1)(d-2)}{2\Lambda}, \\
  A&=\brac{\Phi-\frac{2p-1}{d-2p-1}\frac{\sigma}{r^{\frac{d-2p-1}{2p-1}}}}\dif t, \\
  \sigma^{2p}&=\frac{(d-2)(d-2p-1)}{2^p(2p-1)^2}q^2,\quad c=\frac{2\brac{(d-4)p+1}}{2p-1}.
 \end{align}
\end{subequations}
This is a simple generalisation of the solution considered in \cite{Gonzalez:2009nn} to include a negative cosmological constant, as well as planar and hyperbolic symmetry. As pointed out in \cite{Gonzalez:2009nn}, the asymptotic behaviour of the solution depends on the exponent $p$. For the present purposes, we consider asymptotically-AdS solutions, which is satisfied for $\half<p<\frac{d-1}{2}$. Here the function $f(r)$ also consists of $N_f=4$ terms. Hence by Theorem \ref{thm_unicusps}, its $\mathcal{F}$-$T$ curve should have at most $N_f-2=2$ cusps, similar to the Reissner--Nordstr\"{o}m case. 

The mass, entropy, and charge are given by
\begin{subequations}\label{EpL_MassEntropy}
\begin{align}
 M&=\frac{\Omega}{16\pi G}(d-2)\mu=\frac{\Omega}{16\pi G}(d-2)\brac{kr_+^{d-3}+\frac{r_+^{d-1}}{\ell^2}+\frac{q^2}{r_+^{\frac{d-2p-1}{2p-1}}}},\\
 S&=\frac{\Omega r_+^{d-2}}{4G},\\
 Q&=\frac{\Omega 2^{p-1}p}{4\pi G}\sigma^{2p-1}.
\end{align}
\end{subequations} 
Using the Euclidean periodicity trick, the temperature is 
\begin{align}
 T&=\frac{f'(r_+)}{4\pi}=\frac{1}{4\pi r_+}\sbrac{(d-3)k+(d-1)\frac{r_+^2}{\ell^2}-\brac{\frac{d-2p-1}{2p-1}}q^2r_+^{-\brac{\frac{2((d-4)p+1)}{2p-1}}}}.\label{EpL_T}
\end{align}
For which one can obtain the free energy 
\begin{align}
 \mathcal{F}=M-TS=\frac{\Omega}{16\pi G}\sbrac{kr_+^{d-3}-\frac{r_+^{d-1}}{\ell^2}+\frac{2pd-6p+1}{2p-1}q^2r_+^{-\brac{\frac{d-2p-1}{2p-1}}}}. \label{EpL_F}
\end{align}
This can be verified with a Euclidean action calculation by adding a boundary term $I_Q=\frac{p}{4\pi G}\int\dif^{d-1}y\sqrt{h}\;n_\mu F^{\mu\nu}X^{p-1}A_\nu$ to fix the charge, where $n^\mu$ is the unit normal at the boundary defined by constant $r=r_b$, and the limit $r_b\rightarrow\infty$ is subsequently taken.

Now, to obtain these quantities using the $A$-discriminant, we consider the function $\mathcal{P}=\mathcal{F}-M+TS$. Using the mass and entropy as given by Eq.~\Eqref{EpL_MassEntropy}, we obtain
\begin{align}
 \mathcal{P}&=\mathcal{F}-\frac{\Omega}{16\pi G}(d-2)\brac{kr_+^{d-3}+\frac{r_+^{d-1}}{\ell^2}+\frac{q^2}{r_+^{c-d+3}}}+T\frac{\Omega}{4G}r_+^{d-2}.
\end{align}
However, our algebraic results from $A$-discriminants may not be applicable if $c$ is not an integer. This can be circumvented if $p$ is rational, which in turn makes $c$ rational as well. Then one can perform a coordinate transformation on $r_+$ which converts the expressions into integer powers of the new variable. Indeed, this transformation depends on the specific value of $p$ and $d$. 

As a demonstration, let us consider concretely the case 
\begin{align*}
 d=4,\quad k=1,\quad p=\frac{4}{5}.
\end{align*}
For this value of $p$, $c=\frac{10}{3}$. Introducing dimensionless variables 
\begin{align}
 r_+=x^3\ell,\quad q=w\ell^{5/3},
\end{align}
the function $\mathcal{P}=\mathcal{F}-M+TS$ can written as 
\begin{align}
 \frac{4G}{\ell}x^{7}\mathcal{P}=P(x)=bx^7-2\brac{x^{10}+x^{16}+w^2}+ax^{13},
\end{align}
where 
\begin{align}
 a=4\pi\ell T,\quad b=\frac{4G}{\ell}\mathcal{F}.
\end{align}
Finding the discriminant by $P(x)=P'(x)=0$, we get 
\begin{align}
 a=\frac{3x^{10}+9x^{16}-7w^2}{3x^{13}},\quad b=\frac{3x^{10}-3x^{16}+13w^2}{3x^7}.
\end{align}
Restoring dimensionful variables, 
\begin{align}
 T=\frac{3\ell^2r_+^{10/3}+9r_+^{16/3}-7q^2\ell^2}{12\pi\ell^2r_+^{13/3}},\quad\mathcal{F}=\frac{3\ell^2r_+^{10/3}-3r_+^{16/3}+13q^2\ell^2}{12G\ell^2r_+^{7/3}}.
\end{align}
\begin{figure}
 \centering 
 \includegraphics{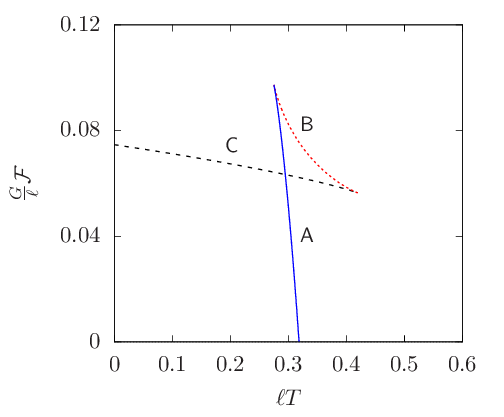}
 \caption{Graph of free energy $\mathcal{F}$ vs temperature $T$ for the Einstein-power-Maxwell in the case $k=1$, $d=4$, $p=\frac{4}{5}$, and $q=0.015\ell^{5/3}$. Similar to the Reissner--Nordstr\"{o}m case, there is a distinctive swallowtail figure with two cusps subdividing the curve into three segments, \textsf{A} ($r>r_2$), \textsf{B} ($r_1<r<r_2$), and \textsf{C} ($0<r<r_1$). In this case, $r_1=0.1652858553\ell$ and $r_2=0.5731353595\ell$.}
 \label{fig_EpM}
\end{figure}
This agrees with the standard methods of Eqs.~\Eqref{EpL_T} and \Eqref{EpL_F}.

As expected, the $\mathcal{F}$-$T$ curve has at most two cusps, similar to the Reissner--Nordstr\"{o}m case. Indeed, this is seen explicitly in Fig.~\ref{fig_EpM} for $k=1$, $d=4$, $p=\frac{4}{5}$, and $q=0.015\ell^{5/3}$. The two cusps correspond to horizon radii 
\begin{align}
 r_1=0.1652858553\ell,\quad r_2=0.5731353595\ell.
\end{align}
The cusps separate the $\mathcal{F}$-$T$ curve into three segments, namely \textsf{A} ($r_+>r_2$), \textsf{B} ($r_1<r_+<r_2$), and \textsf{C} ($0<r_+<r_1$).

\subsection{Euler--Heisenberg black hole}

We now consider a black hole solution in Euler--Heisenberg's model of non-linear electrodynamics \cite{Heisenberg:1936nmg}, whose action is described by
\begin{align}
 I=\frac{1}{16\pi G}\int\dif^4x\sqrt{-g}\brac{R-2\Lambda-F_{\mu\nu}F^{\mu\nu}+\bar{a}\brac{F_{\mu\nu}F^{\mu\nu}}^2+\bar{b}\brac{F_{\mu\nu}*F^{\mu\nu}}^2},
\end{align}
where $*F_{\mu\nu}=\half\sqrt{-g}\varepsilon_{\mu\nu\rho\sigma}F^{\mu\nu}$ is the dual of $F_{\mu\nu}$, with $\sqrt{-g}\varepsilon_{\mu\nu\rho\sigma}$ being the Levi-Civita tensor. The constants $\bar{a}$ and $\bar{b}$ parametrises the intensity of the non-linear terms of the electromagnetic field.\footnote{In our notation, $\bar{a}=\frac{a}{8}$ and $\bar{b}=\frac{b}{8}$ where $a$ and $b$ are the parameters used by \cite{Magos:2020ykt}.} In Euler--Heisenberg theory \cite{Heisenberg:1936nmg}, we take $\bar{b}=\frac{7}{4}\bar{a}$, with $\bar{a}=\frac{\alpha^2}{45m_{\mathrm{e}}^2}$ where $\alpha$ is the fine structure constant and $m_{\mathrm{e}}$ is the electron mass.

In \cite{Magos:2020ykt}, Magos and Breton considered an asymptotically AdS solution
\begin{subequations}
\begin{align}
 \dif s^2&=-f(r)\dif t^2+f(r)^{-1}\dif r^2+r^2\brac{\dif\theta^2+\sin^2\theta\,\dif\phi^2},\\
 A&=\brac{\Phi-\frac{q}{r}+\frac{4q^3\bar{a}}{5r^5}}\dif t,\\
 f(r)&=1-\frac{2m}{r}+\frac{r^2}{\ell^2}+\frac{q^2}{r^2}-\frac{2q^4\bar{a}}{5r^6}.
\end{align}
\end{subequations}
Here the function $f(r)$ consists of $N_f=5$ terms. Therefore by Theorem \ref{thm_unicusps}, the $\mathcal{F}$-$T$ curve for this system may have at most $N_f-2=3$ cusps.

To obtain $\mathcal{F}$ and $T$ using the $A$-discriminant, we take the mass, entropy, and charge which are given by
\begin{subequations}
\begin{align}
 M&=\frac{m}{G}=\frac{r_+}{2G}\brac{1+\frac{r_+^2}{\ell^2}+\frac{q^2}{r_+^2}-\frac{2q^4\bar{a}}{5r_+^6}},\\
 S&=\frac{\pi r_+^2}{G},\\
 Q&=\frac{q}{G}.
\end{align}
\end{subequations}
Defining dimensionless variables
\begin{align}
 x=\frac{r_+}{\ell},\quad w=\frac{q}{\ell},\quad \bar{a}=\sigma\ell^2,
\end{align}
the function $\mathcal{P}=\mathcal{F}-M+TS$ can be cast into the form
\begin{align}
 \frac{10Gx^5}{\ell}\mathcal{P}=10\brac{\frac{\mathcal{F}G}{\ell}}x^5-5x^6-5x^8-5w^2x^4+2w^4\sigma +10\brac{\pi\ell T}x^7
\end{align}
As before, we seek its discriminant by imposing $P(x)=P'(x)=0$, which is solved for $a$ and $b$ as functions of $x$. Restoring dimensionful variables,
\begin{align}
 T=\frac{1}{4\pi r_+}\brac{1+3\frac{q^2}{\ell^2}-\frac{q^2}{r_+^2}+\frac{2q^4\bar{a}}{r_+^6}},\quad \mathcal{F}=\frac{r_+}{4G}\brac{1-\frac{r_+^2}{\ell^2}+3\frac{q^2}{r_+^2}-\frac{14q^4\bar{a}}{5r_+^6}}.
\end{align}
This curve is plotted in Fig.~\ref{fig_EHAdS} for values $q=0.09667\ell$ and $\bar{a}=0.0007\ell^2$. For these values, the curve has three cusps, as expected from Theorem \ref{thm_unicusps}, which separates the curve into three segments, \textsf{A} ($r_+>r_3$), \textsf{B} ($r_2<r_+<r_3$), \textsf{C} ($r_1<r_+<r_2$), and \textsf{D} ($0<r_+<r_1$).

\begin{figure}
 \centering
 \includegraphics{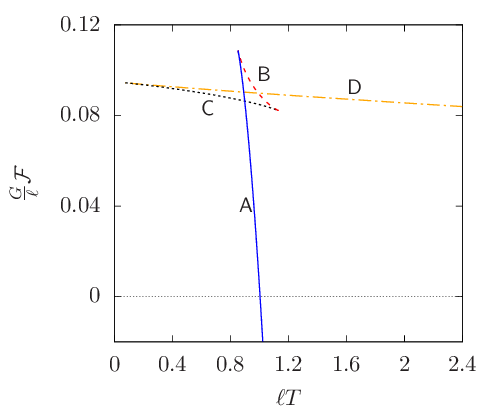}
 \caption{Graph of free energy $\mathcal{F}$ vs $T$ for the Euler--Heisenberg black hole with $q=0.09667\ell$ and $\bar{a}=0.0007\ell^2$. There are three cusps subdividing the curve into four segments, \textsf{A} ($r>r_3$), \textsf{B} ($r_2<r<r_3$), \textsf{C} ($r_1<r<r_2$), and \textsf{D} ($0<r<r_1$), where $r_1$, $r_2$, and $r_3$ are points at which $T'(r_+)=0$. In this example, $r_1=0.07906319692\ell$, $r_2=0.1723732923\ell$, and $r_3=0.5499523272\ell$.}
 \label{fig_EHAdS}
\end{figure}

\subsection{Non-Linear Einstein-Power-Maxwell}
In this example, we consider a non-linear and generalised version of the Einstein-power-Maxwell theory studied in \cite{Gao:2021kvr}. The action is
\begin{equation}
    I = \frac{1}{16\pi G}\int d^4x \sqrt{-g}\brac{R-2\Lambda - \sum^{\infty}_{i=1}\alpha_{i}(F_{\mu\nu}F^{\mu\nu})^i},
\end{equation}
The $\alpha_{i}$ are the \emph{dimensional coupling constants} where $\left[\alpha_{i}\right] = L^{2(i-1)}$.
The equations of motion are
\begin{equation}
\begin{split}
    R_{\mu\nu}-\half R g_{\mu\nu} = -2\frac{dL_{\text{EM}}}{dF^2}F_{\mu}^{\phantom{a}\lambda}F_{\nu\lambda} &+ \frac{1}{2}g_{\mu\nu}L_{\text{EM}},\\
    \nabla_{\mu}\left(\frac{dL_{\text{EM}}}{dF^2}F^{\mu\nu}\right) &= 0,
\end{split}
\end{equation}
where $L_{\text{EM}} = -\sum^{\infty}_{i=1}\alpha_{i}(F_{\mu\nu}F^{\mu\nu})^i$.

The solution found in \cite{Tavakoli:2022kmo} takes the form
\begin{equation}
    \Phi(r) = \sum^{\infty}_{i=1}b_{i}r^{-i} \text{\phantom{aba}and\phantom{aba}}f(r) = 1 + \sum^{\infty}_{i=1}c_{i}r^{-i} + \frac{r^2}{\ell^2},
\end{equation}
with $\alpha_1=1$,
\begin{subequations}
    \begin{align*}
        b_{1} &= Q, \phantom{aba}b_{5} = \frac{4}{5}Q^3\alpha_{2}, \phantom{aba}b_{9} = \frac{4}{3}Q^5\left(4\alpha^2_{2} - \alpha_{3}\right)\\
        b_{13} &= \frac{32}{13}Q^7\left(24\alpha_{2}^3 - 12\alpha_{3}\alpha_{2}+\alpha^4\right),\\
        b_{17} &= \frac{80}{17}Q^9 \left(176\alpha_{2}^4 - 132\alpha_{2}^2 \alpha_{3} + 16\alpha_{4}\alpha_{2}+9\alpha_{3}^2-\alpha_{5}\right),\\
        b_{21} &= \frac{64}{7}Q^{11}\left(1456\alpha_{2}^5+234\alpha_{3}^2\alpha_{2}+208\alpha_{4}\alpha^2_{2}-24\alpha_{4}\alpha_{3}-1456\alpha_{2}^3\alpha_{3}-20\alpha_{5}\alpha_{2}+\alpha_{6}\right),\\
        b_{25} &= \frac{448}{25}Q^{13}(13056\alpha_{2}^6 + 2560\alpha_{2}^3\alpha_{4}-720\alpha_{3}\alpha_{2}\alpha_{4}+16\alpha_{4}^2-300\alpha_{2}^2\alpha^5 \\
        &+ 4320 \alpha_{3}^2\alpha_{2}^2 - 16320\alpha_{2}^4\alpha_{3}+24\alpha_{6}\alpha_{2}-135\alpha^3_{3}+30\alpha_{3}\alpha_{5}-\alpha_{7}),\\
        &\vdots
    \end{align*}
\end{subequations}
and
\begin{equation*}
    c_{1} = -2GM, \phantom{abab}c_{i} = \frac{4Q}{i+2}b_{i-1}
\end{equation*}
for $i > 1$ and that $M$ is the mass of the black hole. Note that if we let all $\alpha_{i} = 0 \phantom{a}\forall i >1$, then we necessarily recover the Reissner--Nordstr\"{o}m-AdS black hole with two horizons.

Writing out the expressions for $\Phi(r)$ and $f(r)$ explicitly, we have
\begin{subequations}
    \begin{align}
        \Phi(r) &= \frac{Q}{r} + \frac{b_{5}}{r^5} + \frac{b_{9}}{r^9} + \frac{b_{13}}{r^{13}} + \frac{b_{17}}{r^{17}} + \frac{b_{21}}{r^{21}} + \frac{b_{25}}{r^{25}},\\
        f(r) &= 1 - \frac{2GM}{r} + \frac{Q^2}{r^2} + \frac{b_{5}Q}{2r^6}+\frac{b_{9}Q}{3r^{10}} + \frac{b_{13}Q}{4r^{14}} + \frac{b_{17}Q}{5r^{18}} + \frac{b_{21}Q}{6r^{22}} + \frac{b_{25}Q}{7r^{26}} + \frac{r^2}{\ell^2}.
    \end{align}
\end{subequations}
Since $f(r)$ has 10 terms, we have $N_{f} =10$. This immediately implies that the $\mathcal{F}$-$T$ curve has at most $N_f-2=8$ cusps according to Theorem \ref{thm_unicusps}.

The mass, entropy, and pressure are given by:
\begin{subequations}\label{non_lin_pressure}
    \begin{align}
        M &= \frac{r_{+}}{2}\brac{1  + \frac{Q^2}{r_{+}^2} + \frac{b_{5}Q}{2r_{+}^6}+\frac{b_{9}Q}{3r_{+}^{10}} + \frac{b_{13}Q}{4r_{+}^{14}} + \frac{b_{17}Q}{5r_{+}^{18}} + \frac{b_{21}Q}{6r_{+}^{22}} + \frac{b_{25}Q}{7r_{+}^{26}} + \frac{r_{+}^2}{\ell^2}},\\
        S &= \frac{\pi}{G} r^2_{+},\\
        P &= \frac{3}{8G\pi\ell^2},
    \end{align}
\end{subequations}
where $r=r_{+}$ is the horizon, defined as the largest root of $f(r)=0$.

Using the Euclidean periodicity trick, the temperature is
\begin{equation}
    T = \frac{1}{4\pi}\frac{df(r_{+})}{dr} = \frac{1}{4\pi r_{+}}\left(1+\frac{3r^2_{+}}{\ell^2}-\frac{Q^2}{r^2_{+}}-\frac{5b_{5}Q}{2r^6_{+}}-\frac{3b_{9}Q}{r^{10}_{+}}-\frac{13b_{13}Q}{4r^{14}_{+}}-\frac{17b_{17}Q}{5r_{+}^{18}}-\frac{7b_{21}Q}{2r_{+}^{22}}-\frac{25b_{25}Q}{7r_{+}^{26}}\right).
\end{equation}
Correspondingly, the free energy is
\begin{align}
 \mathcal{F} &= M - TS\nonumber\\
    &= \frac{1}{G}\Biggl(-\frac{1}{4\ell^2}r_+^3 + \frac{1}{4}r_+ + \frac{3Q^2}{4}\frac{1}{r_+}+\frac{7b_{5}Q}{8}\frac{1}{r_+^5}+\frac{11b_{9}Q}{12}\frac{1}{r_+^9}+\frac{15b_{13}Q}{16}\frac{1}{r_+^{13}}+\frac{19b_{17}Q}{20}\frac{1}{r_+^{17}}\nonumber\\
    &\hspace{2cm}+\frac{23b_{21}Q}{24}\frac{1}{r_+^{21}}+\frac{27b_{25}Q}{28}\frac{1}{r_+^{25}}\Biggl).
\end{align}
Taking numerical values for the parameters on the paper by \cite{Tavakoli:2022kmo}. In particular, these values are:
\begin{equation}\label{non_lin_parameter}
    \begin{split}
        P &= 7.82 \times 10^{-5}, \phantom{aba}Q = 6.623359974,\phantom{aba}\alpha_{2} = -21.63694203,\\
        \alpha_{3} &= 1493.535254, \phantom{aba}\alpha_{4} = -148046.3896,\phantom{aba}\alpha_{5} = 1.759261993 \times 10^7,\\
        \alpha_{6} &= -2.332423991 \times 10^9, \phantom{aba}\alpha_{7} = 3.3277815589 \times 10^{11},
    \end{split}
\end{equation}

By solving $T'(r_+) = 0$ numerically, the location of the cusps $r_i$ (with $T > 0$) was found to be:
\begin{equation}\label{non_lin_EM_horizon}
  \begin{split}
      r_{1} &= 0.1280105460\ell,\phantom{aba}
      r_{2}  = 0.1407566044\ell,\\
      r_{3} &= 0.1791671242\ell,\phantom{aba}
      r_{4}  = 0.2175688113\ell,\\
      r_{5} &= 0.3326818758\ell,\phantom{ba}
      r_{6}  = 0.4352187767\ell,
  \end{split}
\end{equation}
where $\ell=\sqrt{\frac{3}{4\pi GP}}$ can be computed using the given value of $P$ in Eqn.~\Eqref{non_lin_parameter}.

\begin{figure}
    \centering
    \includegraphics[scale=0.80]{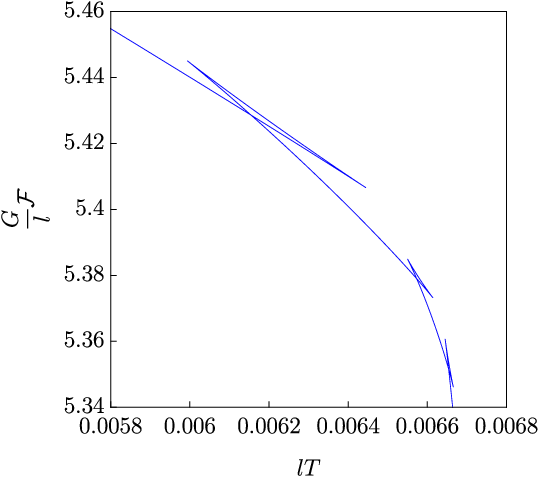}
    \caption{Graph of free energy $\mathcal{F}$ against temperature $T$ for the black hole in non-linear Einstein-Maxwell theory. The parameters used to plot these curves are defined in Eq.\Eqref{non_lin_EM_horizon}.}
    \label{fig_non_lin_EM_main}
\end{figure}

Since $\mathcal{F}\equiv \mathcal{F}(r_{+})$ and $T\equiv T(r_{+})$, we can plot $\mathcal{F}$ against $T$ as a parameterised curve $(T(r_{+}),\mathcal{F}(r_{+}))$. The aforementioned plot is shown in Fig. \ref{fig_non_lin_EM_main}. The main feature of this plot is that there are \textit{three} swallowtail figures, each of which has \textit{two} cusps subdividing the curve into three segments. There are six cusps in total. These three separate swallowtails correspond to three phase stable phase transitions. As shown in \cite{Tavakoli:2022kmo}, it is also possible to have four phases, i.e. eight cusps, in this case. This saturates Theorem \ref{thm_unicusps}.

In this particular case where we have three phases, the three swallowtails can arranged to meet at the same temperature by slight adjustments in the extrema. Plotting the $P-T$ curve shows the existence of a \textit{quadruple point} \cite{Tavakoli:2022kmo}. Therefore, this suggests that for an $n$-tuple point to form, one needs $n-1$ swallowtails, which is equivalent to $n-1$ first order phase transitions.

By further adjustments to the extrema, the three swallowtails can arranged to meet at the same temperature, and this was done in \cite{Tavakoli:2022kmo}. Briefly discussed in \cite{Tavakoli:2022kmo} was the prospect of having an extra swallowtail figure by adding two coupling constants. Introducing two new coupling constants $b_{29}$ and $b_{33}$, $f(r)$ now reads:
\begin{equation}
    f(r) = 1 - \frac{2GM}{r} + \frac{Q^2}{r^2} + \frac{b_{5}Q}{2r^6}+\frac{b_{9}Q}{3r^{10}} + \frac{b_{13}Q}{4r^{14}} + \frac{b_{17}Q}{5r^{18}} + \frac{b_{21}Q}{6r^{22}} + \frac{b_{25}Q}{7r^{26}} + \frac{b_{29}Q}{8r^{30}} + \frac{b_{33}Q}{9r^{34}}+\frac{r^2}{\ell^2}.
\end{equation}
By a judicious relabelling of parameters, we can write $f(r)$ in a more succinct way:
\begin{equation}
    f(r) = 1 - \frac{2GM}{r} + \frac{Q^2}{r^2} + \frac{B_{2}Q^6}{r^6}+\frac{B_{3}Q^{10}}{r^{10}} + \frac{B_{4}Q^{14}}{r^{14}} + \frac{B_{5}Q^{18}}{r^{18}} + \frac{B_{6}Q^{22}}{r^{22}} + \frac{B_{7}Q^{26}}{r^{26}} + \frac{B_{8}Q^{30}}{r^{30}} + \frac{B_{9}Q^{34}}{r^{34}}+\frac{r^2}{\ell^2}.
\end{equation}
Using these numerical values for the coupling constants and parameters:
\begin{equation}\label{para_quintuple}
    \begin{split}
        P &= 0.0001391912090,\phantom{a}Q = 5.319760291, \phantom{a}
B_2 = -0.3247933480,\\
B_3 &= 0.3841745537, \phantom{a}
B_4 = -0.4189261389,\phantom{a}
B_5 = 0.3278972499,\\
B_6 &= -0.1613462070, \phantom{a}
B_7 = 0.04255971780, \phantom{a}
B_8 = -0.004407159679,\\
B_9 &= 0.0001534401794.
    \end{split}
\end{equation}
Proceeding as above, we find that the graph of free energy $\mathcal{F}$ against $T$ has eight cusps. The plot was shown in Figure. \ref{fig_quintuple}.

By solving $T'(r_{+}^{(i)}) = 0$ numerically, the cusps are located at:
\begin{equation}
    \begin{split}
        r_{1} &= 0.18064291450\ell, \phantom{ab}r_{2} = 0.18781399271\ell,\\r_{3} &= 0.19976574497\ell, \phantom{ab}r_{4} = 0.218378150812\ell,\\
        r_{5} &= 0.24415812831\ell, \phantom{ab}r_{6} = 0.27830477361\ell,\\
        r_{7} &= 0.31757635308\ell, \phantom{ab}r_{8} = 0.35445618930\ell.
    \end{split}
\end{equation}
According to Theorem \ref{thm_unicusps}, the maximum number of cusps allowed is $N_{f} - 2 = 10$. In this case, we have eight cusps, so the theorem is once again satisfied.

In general, for a given odd value of $N=2n-1$, where $N$ denotes the number of coupling constants $\alpha_{i}$ in the theory, there will be at most an $(n+1)$-critical point \cite{Tavakoli:2022kmo}. This translates to the fact that there will be at most $2n$ cusps.

In this theory, we have that $N_{f} =N+3$. Using Theorem \ref{thm_unicusps}, we can see that the maximum number of cusps is:
\begin{equation}
    N_{f} -2 = N+3 -2 = 2n-1+3-2 = 2n.
\end{equation}
Therefore, Theorem \ref{thm_unicusps} is never violated in this theory.

Therefore, we expect $N=5$ to produce a quadruple point, similar to Figure. \ref{fig_non_lin_EM_main}. In fact, this was done in \cite{Tavakoli:2022kmo}. In this case, $N_{f} = 8$, so the maximum number of cusps is $6$. Since the quadruple point consists of six cusps, Theorem \ref{thm_unicusps} is saturated.

\begin{figure}
    \centering
    \includegraphics[scale=0.80]{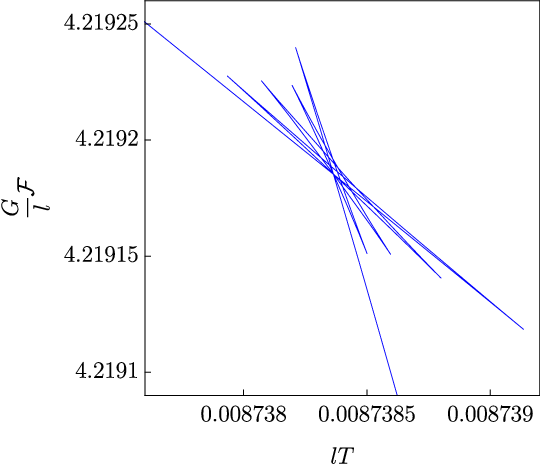}
    \caption{Graph of free energy $\mathcal{F}$ against temperature $T$ for the black hole in non-linear Einstein-Maxwell theory with eight cusps. The parameters used to plot these curves are defined in Eq. \Eqref{para_quintuple}.}
    \label{fig_quintuple}
\end{figure}

\section{The grand canonical ensemble} \label{sec_gc}

All the cases considered in the sections above have been studied in canonical (i.e., fixed-charge) ensemble where the relevant thermodynamic potential is the free energy $\mathcal{F}=M-TS$. Then the system is described by a single variable $r_+$ which leads to a single-variable polynoimal $P(x)$ for which we apply Theorem \ref{thm_unicusps}. In the grand canonical ensemble, the charge $q$ is no longer fixed and now plays the role of a second variable after $r_+$. In this case, the electric potential is fixed at the boundary and the relevant thermodynamic potential is \emph{grand potential} $\mathcal{W}=M-TS-\Phi Q$ \cite{Chamblin:1999tk,Chamblin:1999hg,Ghosh:2021uxg}, and the system is now described by two variables $r_+$ and $q$.

In the case of the grand canonical ensemble, Theorem \ref{thm_unicusps} is no longer applicable, and we do not have a bound on the number of cusps. But nevertheless it is still possible to show that this system can still be described by an $A$-discriminant, this time for a bivariate polynomial. Let us demonstrate this by working concretely on the Reissner--Nordstr\"{o}m-AdS black hole in the fixed potential ensemble, where the mass $M$, entropy $S$, and charge $Q$ is given in Eq.~\Eqref{RNAdS_quantities}. As before we let $T$ and $\Phi$ be the temperature and electric potential of the black hole. 

Following the procedure analogous to Sec.~\ref{sec_thermo}, we consider the function $\mathcal{P}=\mathcal{W}-M+TS+\Phi Q$, which vanishes whenever $\mathcal{W}$ takes the value of the grand potential of the system,
\begin{align}
 \mathcal{P}=\mathcal{W}-M+TS+\Phi Q=0.\label{GC_condition1}
\end{align}
Now consider a small change of the two parameters 
\begin{align}
 r_+\rightarrow r_++\delta r_+,\quad q\rightarrow q+\delta q. \label{RNAdS-gc_perturb}
\end{align}
Then the mass, entropy, and charge change according to 
\begin{align}
 M\rightarrow M+\delta M,\quad S\rightarrow S+\delta S,\quad Q\rightarrow Q+\delta Q,
\end{align}
where now
\begin{align}
 \delta M=\frac{\partial M}{\partial r_+}\delta r_++\frac{\partial M}{\partial q}\delta q,\quad \delta S=\frac{\partial S}{\partial r_+}\delta r_+,\quad \delta Q=\frac{\partial Q}{\partial q}\delta q.
\end{align}
In the grand canonical ensemble, we have an additional thermodynamic extensive variable $Q$ and the corresponding first law \cite{Chamblin:1999tk,Chamblin:1999hg}
\begin{align}
 \delta M=T\delta S+\Phi\delta Q
\end{align}
is obeyed. We then rearrange the first law $\delta\mathcal{W}=\delta M+T\delta S+\Phi\delta Q$ as
\begin{align}
 \delta\brac{\mathcal{W}-M+TS+\Phi Q}=\delta\mathcal{P}&=0\nonumber\\
 \frac{\partial\mathcal{P}}{\partial r_+}\delta r_++\frac{\partial\mathcal{P}}{\partial q}\delta q&=0\label{GC_condition2}.
\end{align}
In other words, if we take $\mathcal{P}$ to be a function of two variables $r_+$ and $q$, conditions \Eqref{GC_condition1} and \Eqref{GC_condition2} lead to
\begin{align}
 \mathcal{P}=\frac{\partial\mathcal{P}}{\partial r_+}=\frac{\partial\mathcal{P}}{\partial q}=0.\label{GC_3}
\end{align}
If $\mathcal{P}=P/H$ is a rational function where $P$ is some polynomial in $r_+$ and $q$, then this is equivalent to
\begin{align}
 P=\frac{\partial P}{\partial r_+}=\frac{\partial P}{\partial q}=0,
\end{align}
which defines the $A$-discriminant for a two-variable polynomial $P$.

We now write out $\mathcal{P}=\mathcal{W}-M+TS+\Phi Q$ explicitly for the spherical Reissner--Nordstr\"{o}m-AdS black hole in the grand canonical ensemble. Using Eq.~\Eqref{RNAdS_quantities}, we have 
\begin{align}
 \frac{16\pi G}{\Omega\ell^{d-3}}x^{d-3}\mathcal{P}\equiv P(x,y)&=bx^{d-3}-(d-2)\brac{x^{2d-6}+x^{2d-4}+y^2}+ax^{2d-5}+cyx^{d-3},
\end{align}
where we are using dimensionless variables $x=r_+/\ell$ and $y=q/\ell^{d-3}$, along with 
\begin{align}
 a=4\pi\ell T,\quad b=\frac{16\pi G}{\Omega\ell^{d-3}}\mathcal{W},\quad c=2\sqrt{2(d-2)(d-3)}\,\Phi.
\end{align}
Here $P(x,y)$ is a two-variable polynomial consisting of six($\neq5$) monomials. So neither Theorem \ref{thm_korben} nor \ref{thm_unicusps} is applicable here, and the bounds of the number of cusps is not known. In any case, for this particular example of the Reissner--Nordstr\"{o}m-AdS black hole, we solve $P=\frac{\partial P}{\partial x}=\frac{\partial P}{\partial y}=0$ for $a$, $b$, and $c$. This leads to the parametric equations of the $A$-discriminant in terms of the roots $(x,y)$ 
\begin{align*}
 a&=\frac{1}{x}\sbrac{d-3+(d-1)x^2-(d-3)\frac{y^2}{x^{2d-6}}},\quad b=x^{d-3}-x^{d-1}-\frac{y^2}{x^{d-3}},\quad c=2(d-2)\frac{y}{x^{d-3}}.
\end{align*}
Restoring dimensionful variables, we recover 
\begin{subequations}
\begin{align}
 T&=\frac{1}{4\pi r_+}\sbrac{d-3+(d-1)\frac{r_+^2}{\ell^2}-(d-3)\frac{q^2}{r_+^{2d-6}}},\\
 \mathcal{W}&=\frac{\Omega}{16\pi G}\brac{r_+^{d-3}-\frac{r_+^{d-1}}{\ell^2}-\frac{q^2}{r_+^{d-3}}},\\
 \Phi&=\sqrt{\frac{d-2}{2(d-3)}}\,\frac{q}{r_+^{d-3}},
\end{align}
\end{subequations}
thus recovering the correct expressions for temperature, grand potential, and electric potential in agreement with \cite{Chamblin:1999tk}. In this case, the $A$-discriminant is a two-dimensional hypersurface in the three dimensional space spanned by $(T,\mathcal{W},\Phi)$. To obtain a one-dimensional curve such as in Fig.~4 of \cite{Chamblin:1999tk}, we express $T$ and $\mathcal{W}$ in terms of the potential $\Phi$ by eliminating $q$, giving
\begin{subequations}
\begin{align}
 T&=\frac{1}{4\pi r_+}\sbrac{(d-3)\brac{1-C^2\Phi^2}+(d-1)\frac{r_+^2}{\ell^2}},\\
 \mathcal{W}&=\frac{\Omega}{16\pi G}\sbrac{\brac{1-C^2\Phi^2}r_+^{d-3}-\frac{r_+^{d-1}}{\ell^2}},
\end{align}
\end{subequations}
where $C=\sqrt{\frac{2(d-3)}{d-2}}$. This recovers Eqs.~(22) and (24) of \cite{Chamblin:1999tk}. For example, the curve for $d=4$ and $\Phi=0.7$ is plotted in Fig.~\ref{fig_RNAdS-gc}. Here we have two branches of solutions separated at $r_1=\sqrt{\frac{1-\Phi^2}{3}}$ which is a cusp. Threfore, this single cusps separates the system into two branches, or phases. In Fig.~\ref{fig_RNAdS-gc}, branch \textsf{A} is $r_1<r_+<\infty$ and branch \textsf{B} is $0<r_+<r_1$.

\begin{figure}
 \centering 
 \includegraphics{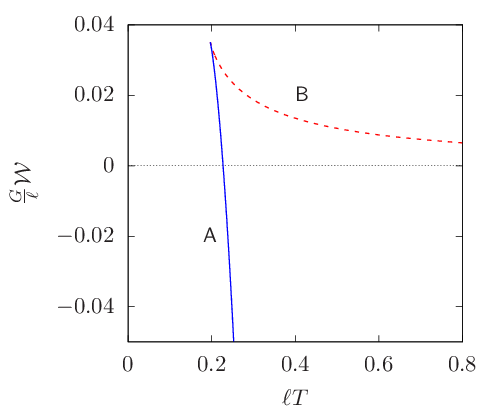}
 \caption{Graph of the grand canonical potential $\mathcal{W}$ vs temperature $T$ for the Reissner--Nordstr\"{o}m-AdS black hole in the fixed potential ensemble, for $d=4$, $\Phi=0.7$. Branch \textsf{A} corresponds to $r_1<r_+<\infty$ and branch \textsf{B} corresponds to $0<r_+<r_1$.}
 \label{fig_RNAdS-gc}
\end{figure}

\section{Conclusion} \label{sec_conclusion}

In this paper, we have studied the cusps of the free energy ($\mathcal{F}$) vs temperature ($T$) curves of various black holes from the perspective of $A$-discriminants. The number of possible distinct thermodynamic phases depends on the number of terms in its function $f(r)=-g_{tt}$, where we consider static black holes of spherical, planar, or hyperbolic symmetry. In particular, if $f(r)$ consists of $N_f$ distinct terms, then its $\mathcal{F}$-$T$ curve has at most $N_f-2$ cusps. This means that there are at most $N_f-1$ distinct thermodynamic phases. 

Note that we have considered the canonical ensemble, where the relevant thermodynamic potential is $\mathcal{F}=M-TS$, and the states are parametrised by a single variable, namely the horizon radius $r_+$. In terms of $A$-discriminants, this translates to considering univariate polynomials where we have established Theorem \ref{thm_unicusps} which gives an upper bound of the number of phases of a given system.

As a final remark, let us see how this present formulation of $A$-discriminants can be connected to other mathematical viewpoints of (gravitational) thermodynamics. It is interesting to note that recently, Refs.~\cite{Xu:2023vyj,Wang:2023qxw} have also applied a mathematical perspective to the phase transitions of black holes, albeit from a different direction. In their case, they considered Riemann surfaces, where it is the winding number that determines the nature of thermodynamic phases. Since this paper essentially studies the same systems using $A$-discriminants, it is tempting to consider the possibility that there might be some connection (or even correspondence) between the winding number of Riemann surfaces and cusps of $A$-discriminants. Besides this, there are other formulations of thermodynamics, such as contact geometry \cite{Rajeev:2007uk,Ghosh:2019rsu,Ghosh:2023khd} or the Riemannian geometry of thermodynamic phase space \cite{Weinhold:1975xej,Ruppeiner:1979bcp,Quevedo:2006xk,Aragon-Munoz:2021hbb}. It might be interesting to explore how an algebraic object like an $A$-discriminant might relate to these geometric settings.

\section*{Acknowledgments}
 
M.~N. is supported by Xiamen University Malaysia Research Fund (Grant No. XMUMRF/2020-C5/IMAT/0013). Y.-K.~L is supported by Xiamen University Malaysia Research Fund (Grant no. XMUMRF/2021-C8/IPHY/0001). L.C. is partially supported by Trinity College, Cambridge via the Rouse Ball Travelling Studentship in Mathematics. Special thanks to Jerry Wu for his invaluable assistance in plotting Figs.~\ref{fig_non_lin_EM_main} and \ref{fig_quintuple}.

\bibliographystyle{BHdiscrim}

\bibliography{BHdiscrim}

\providecommand{\href}[2]{#2}\begingroup\raggedright\begin{thebibliography}{10}

\bibitem{Hawking:1982dh}
S.~W. Hawking and D.~N. Page, {\it {Thermodynamics of Black Holes in anti-De
  Sitter Space}},  Commun. Math. Phys. {\bf 87} (1983) 577.

\bibitem{Witten:1998zw}
E.~Witten, {\it {Anti-de Sitter space, thermal phase transition, and
  confinement in gauge theories}},  Adv. Theor. Math. Phys. {\bf 2} (1998)
  505--532, [\href{http://arxiv.org/abs/hep-th/9803131}{{\tt hep-th/9803131}}].

\bibitem{Chamblin:1999tk}
A.~Chamblin, R.~Emparan, C.~V. Johnson, and R.~C. Myers, {\it {Charged AdS
  black holes and catastrophic holography}},  Phys. Rev. D {\bf 60} (1999)
  064018, [\href{http://arxiv.org/abs/hep-th/9902170}{{\tt hep-th/9902170}}].

\bibitem{Chamblin:1999hg}
A.~Chamblin, R.~Emparan, C.~V. Johnson, and R.~C. Myers, {\it {Holography,
  thermodynamics and fluctuations of charged AdS black holes}},  Phys. Rev. D
  {\bf 60} (1999) 104026, [\href{http://arxiv.org/abs/hep-th/9904197}{{\tt
  hep-th/9904197}}].

\bibitem{Teitelboim:1985dp}
C.~Teitelboim, {\it {The Cosmological Constant as a Thermodynamic Black Hole
  Parameter}},  Phys.~Lett.~B {\bf 158} (1985) 293--297.

\bibitem{Kastor:2009wy}
D.~Kastor, S.~Ray, and J.~Traschen, {\it {Enthalpy and the Mechanics of AdS
  Black Holes}},  Class. Quant. Grav. {\bf 26} (2009) 195011,
  [\href{http://arxiv.org/abs/0904.2765}{{\tt arXiv:0904.2765}}].

\bibitem{Kubiznak:2012wp}
D.~Kubiznak and R.~B. Mann, {\it {P-V criticality of charged AdS black holes}},
   JHEP {\bf 07} (2012) 033, [\href{http://arxiv.org/abs/1205.0559}{{\tt
  arXiv:1205.0559}}].

\bibitem{Kastor:2010gq}
D.~Kastor, S.~Ray, and J.~Traschen, {\it {Smarr Formula and an Extended First
  Law for Lovelock Gravity}},  Class. Quant. Grav. {\bf 27} (2010) 235014,
  [\href{http://arxiv.org/abs/1005.5053}{{\tt arXiv:1005.5053}}].

\bibitem{Kastor:2011qp}
D.~Kastor, S.~Ray, and J.~Traschen, {\it {Mass and Free Energy of Lovelock
  Black Holes}},  Class. Quant. Grav. {\bf 28} (2011) 195022,
  [\href{http://arxiv.org/abs/1106.2764}{{\tt arXiv:1106.2764}}].

\bibitem{Kastor:2016bph}
D.~Kastor, S.~Ray, and J.~Traschen, {\it {Extended First Law for Entanglement
  Entropy in Lovelock Gravity}},  Entropy {\bf 18} (2016), no.~6 212,
  [\href{http://arxiv.org/abs/1604.04468}{{\tt arXiv:1604.04468}}].

\bibitem{Kubiznak:2016qmn}
D.~Kubiznak, R.~B. Mann, and M.~Teo, {\it {Black hole chemistry: thermodynamics
  with Lambda}},  Class. Quant. Grav. {\bf 34} (2017), no.~6 063001,
  [\href{http://arxiv.org/abs/1608.06147}{{\tt arXiv:1608.06147}}].

\bibitem{Kubiznak:2014zwa}
D.~Kubiznak and R.~B. Mann, {\it {Black hole chemistry}},  Can. J. Phys. {\bf
  93} (2015), no.~9 999--1002, [\href{http://arxiv.org/abs/1404.2126}{{\tt
  arXiv:1404.2126}}].

\bibitem{Kastor:2014dra}
D.~Kastor, S.~Ray, and J.~Traschen, {\it {Chemical Potential in the First Law
  for Holographic Entanglement Entropy}},  JHEP {\bf 11} (2014) 120,
  [\href{http://arxiv.org/abs/1409.3521}{{\tt arXiv:1409.3521}}].

\bibitem{Kastor:2018cqc}
D.~Kastor, S.~Ray, and J.~Traschen, {\it {Black Hole Enthalpy and Scalar
  Fields}},  Class. Quant. Grav. {\bf 36} (2019), no.~2 024002,
  [\href{http://arxiv.org/abs/1807.09801}{{\tt arXiv:1807.09801}}].

\bibitem{Magos:2020ykt}
D.~Magos and N.~Bret\'on, {\it {Thermodynamics of the Euler-Heisenberg-AdS
  black hole}},  Phys. Rev. D {\bf 102} (2020), no.~8 084011,
  [\href{http://arxiv.org/abs/2009.05904}{{\tt arXiv:2009.05904}}].

\bibitem{Tavakoli:2022kmo}
M.~Tavakoli, J.~Wu, and R.~B. Mann, {\it {Multi-critical points in black hole
  phase transitions}},  JHEP {\bf 12} (2022) 117,
  [\href{http://arxiv.org/abs/2207.03505}{{\tt arXiv:2207.03505}}].

\bibitem{Quijada:2023fkc}
C.~Quijada, A.~Anabal\'on, R.~B. Mann, and J.~Oliva, {\it {Triple Points of
  Gravitational AdS Solitons and Black Holes}},
  \href{http://arxiv.org/abs/2308.16341}{{\tt arXiv:2308.16341}}.

\bibitem{GKZ-94}
{I.~M.~Gelfand, M.~M.~Kapranov, and A.~V.~Zelevinsky}, {\em {Discriminants,
  Resultants, and Multidimensional Determinants}}.
\newblock {Birkh\"{a}user}, {Boston}, (1994).

\bibitem{Dickenstein2007}
A.~Dickenstein, J.~M. Rojas, K.~Rusek, and J.~Shih, {\it {Extremal Real
  Algebraic Geometry and A-Discriminants}},  Moscow Mathematical Journal {\bf
  7} (2007), no.~3 425--452.

\bibitem{Rusekthesis}
{K.~A.~Rusek}, {\em {A-Discriminant varieties and amoebae}}.
\newblock PhD thesis, {Texas A\&M University}, {College station, Texas},
  {2013}.

\bibitem{Ghosh:2021uxg}
A.~Ghosh, S.~Mukherji, and C.~Bhamidipati, {\it {Logarithmic corrections to the
  entropy function of black holes in the open ensemble}},  Nucl. Phys. B {\bf
  982} (2022) 115902, [\href{http://arxiv.org/abs/2104.12720}{{\tt
  arXiv:2104.12720}}].

\bibitem{Jacobson:2007tj}
T.~Jacobson, {\it {When is $g_{tt} g_{rr} = -1$?}},  Class. Quant. Grav. {\bf
  24} (2007) 5717--5719, [\href{http://arxiv.org/abs/0707.3222}{{\tt
  arXiv:0707.3222}}].

\bibitem{Padmanabhan:2002sha}
T.~Padmanabhan, {\it {Classical and quantum thermodynamics of horizons in
  spherically symmetric space-times}},  Class. Quant. Grav. {\bf 19} (2002)
  5387--5408, [\href{http://arxiv.org/abs/gr-qc/0204019}{{\tt gr-qc/0204019}}].

\bibitem{Hartle:1976tp}
J.~B. Hartle and S.~W. Hawking, {\it {Path Integral Derivation of Black Hole
  Radiance}},  Phys. Rev. D {\bf 13} (1976) 2188--2203.

\bibitem{Gibbons:1976pt}
G.~W. Gibbons and M.~J. Perry, {\it {Black Holes and Thermal Green's
  Functions}},  Proc. Roy. Soc. Lond. A {\bf 358} (1978) 467--494.

\bibitem{Gibbons:1976ue}
G.~W. Gibbons and S.~W. Hawking, {\it {Action Integrals and Partition Functions
  in Quantum Gravity}},  Phys. Rev. D {\bf 15} (1977) 2752--2756.

\bibitem{Lewkowycz:2013nqa}
A.~Lewkowycz and J.~Maldacena, {\it {Generalized gravitational entropy}},  JHEP
  {\bf 08} (2013) 090, [\href{http://arxiv.org/abs/1304.4926}{{\tt
  arXiv:1304.4926}}].

\bibitem{Rajeev:2007uk}
S.~G. Rajeev, {\it {A Hamilton-Jacobi Formalism for Thermodynamics}},  Annals
  Phys. {\bf 323} (2008) 2265--2285,
  [\href{http://arxiv.org/abs/0711.4319}{{\tt arXiv:0711.4319}}].

\bibitem{Ghosh:2019rsu}
A.~Ghosh and C.~Bhamidipati, {\it {Contact Geometry and Thermodynamics of Black
  Holes in AdS Spacetimes}},  Phys. Rev. D {\bf 100} (2019), no.~12 126020,
  [\href{http://arxiv.org/abs/1909.11506}{{\tt arXiv:1909.11506}}].

\bibitem{Gonzalez:2009nn}
H.~A. Gonzalez, M.~Hassaine, and C.~Martinez, {\it {Thermodynamics of charged
  black holes with a nonlinear electrodynamics source}},  Phys. Rev. D {\bf 80}
  (2009) 104008, [\href{http://arxiv.org/abs/0909.1365}{{\tt
  arXiv:0909.1365}}].

\bibitem{Roychowdhury:2012vj}
D.~Roychowdhury, {\it {AdS/CFT superconductors with Power Maxwell
  electrodynamics: reminiscent of the Meissner effect}},  Phys. Lett. B {\bf
  718} (2013) 1089--1094, [\href{http://arxiv.org/abs/1211.1612}{{\tt
  arXiv:1211.1612}}].

\bibitem{Heisenberg:1936nmg}
W.~Heisenberg and H.~Euler, {\it {Consequences of Dirac's theory of
  positrons}},  Z. Phys. {\bf 98} (1936), no.~11 714--732,
  [\href{http://arxiv.org/abs/physics/0605038}{{\tt physics/0605038}}].

\bibitem{Gao:2021kvr}
C.~Gao, {\it {Black holes with many horizons in the theories of nonlinear
  electrodynamics}},  Phys. Rev. D {\bf 104} (2021), no.~6 064038,
  [\href{http://arxiv.org/abs/2106.13486}{{\tt arXiv:2106.13486}}].

\bibitem{Xu:2023vyj}
Z.-M. Xu, Y.-S. Wang, B.~Wu, and W.-L. Yang, {\it {Generalized Maxwell equal
  area law and black holes in complex free energy}},  Phys. Lett. B {\bf 850}
  (2024) 138528, [\href{http://arxiv.org/abs/2305.05916}{{\tt
  arXiv:2305.05916}}].

\bibitem{Wang:2023qxw}
Y.-S. Wang, Z.-M. Xu, and B.~Wu, {\it {Thermodynamic phase transition and
  winding number for the third-order Lovelock black hole}},
  \href{http://arxiv.org/abs/2307.01569}{{\tt arXiv:2307.01569}}.

\bibitem{Ghosh:2023khd}
A.~Ghosh and C.~Bhamidipati, {\it {Contact and metric structures in black hole
  chemistry}},  Front. in Phys. {\bf 11} (2023) 1132712,
  [\href{http://arxiv.org/abs/2302.04467}{{\tt arXiv:2302.04467}}].

\bibitem{Weinhold:1975xej}
F.~Weinhold, {\it {Metric geometry of equilibrium thermodynamics}},  J. Chem.
  Phys. {\bf 63} (1975), no.~6 2479.

\bibitem{Ruppeiner:1979bcp}
G.~Ruppeiner, {\it {Thermodynamics: A Riemannian geometric model}},  Phys. Rev.
  A {\bf 20} (1979), no.~4 1608.

\bibitem{Quevedo:2006xk}
H.~Quevedo, {\it {Geometrothermodynamics}},  J. Math. Phys. {\bf 48} (2007)
  013506, [\href{http://arxiv.org/abs/physics/0604164}{{\tt physics/0604164}}].

\bibitem{Aragon-Munoz:2021hbb}
L.~Aragon-Munoz and H.~Quevedo, {\it {Symplectic structure of equilibrium
  thermodynamics}},  Int. J. Geom. Meth. Mod. Phys. {\bf 19} (2022), no.~11
  2250178, [\href{http://arxiv.org/abs/2104.13009}{{\tt arXiv:2104.13009}}].

\end{thebibliography}\endgroup

\end{document}